\documentclass[11pt]{article}

\usepackage[numbers]{natbib}

\usepackage[ruled,vlined]{algorithm2e}

\SetAlFnt{\small}
\SetAlCapFnt{\small}
\SetAlCapNameFnt{\small}
\SetAlCapHSkip{0pt}
\IncMargin{-\parindent}

\usepackage[letterpaper, width=6.8in, height=9.3in, marginratio={1:1, 1:1}]{geometry}

\usepackage{fullpage}
\usepackage{float}
\usepackage{times}
\usepackage{epsfig}
\usepackage{mathtools}
\usepackage{amssymb}
\usepackage{amstext}
\usepackage{amsthm}
\usepackage{xspace}
\usepackage{latexsym}
\usepackage{verbatim}
\usepackage{multirow}
\usepackage{graphicx}
\usepackage{epstopdf} 
\usepackage{ifthen}
\usepackage{subcaption}
\usepackage{enumitem}
\usepackage[dvipsnames,usenames]{xcolor}
\usepackage{prettyref}
\usepackage[hypertexnames=false,colorlinks=true,pdfpagemode=Usenone,urlcolor=Blue,linkcolor=RoyalBlue,citecolor=OliveGreen,pdfstartview=FitH]{hyperref}

\usepackage{color-edits}
\addauthor{shuchi}{red}	
\addauthor{nikhil}{blue}
\addauthor{anna}{wildstrawberry}
\addauthor{balu}{green}

\newcommand{\AutoAdjust}[3]{\mathchoice{ \left #1 #2  \right #3}{#1 #2 #3}{#1 #2 #3}{#1 #2 #3} }

\newcommand{\InBrackets}[1]{\AutoAdjust{[}{#1}{]}}
\newcommand{\Ex}[2][]{\operatorname{\mathbf E}_{#1}\InBrackets{#2}}

\newcommand{\Prx}[2][]{\operatorname{\mathbf{P}}_{#1}\InBrackets{#2}}

\newcommand{\fall}[1]{\underbar{r}}

\newcommand{\argmax}{\operatorname{argmax}}

\newcommand{\price}{p}
\newcommand{\prices}{{\mathbf \price}}

\newcommand{\dist}{F}

\newcommand{\dens}{f}

\newcommand{\ppp}{\ensuremath{\mathrm{PPP}}}
\newcommand{\bin}{\ensuremath{\mathrm{BIN}}}

\newcommand{\mnp}{\mu}
\newcommand{\mye}{\mathcal{M}}
\newcommand{\rmye}{\mathcal{R}^{\scriptscriptstyle \mye}}

\newcommand{\rppp}[1][0]{\mathcal{R}^{\scriptscriptstyle \ppp}_{#1}}
\newcommand{\rbin}[1][0]{\mathcal{R}^{\scriptscriptstyle \bin}_{#1}}

\newcommand{\aon}{binary value}
\newcommand{\Aon}{Binary value}
\newcommand{\rw}{random walk}

\newcommand{\st}{T(V_0)}
\newcommand{\cv}{\mathcal{C}}

\newcommand{\hab}[2]{h_{#1,#2}}
\newcommand{\habc}[3]{h_{#1,\{#2,#3\}}}

\newtheorem{theorem}{Theorem}[section]
\newtheorem{definition}{Definition}

\newtheorem{lemma}[theorem]{Lemma}

\newtheorem{corollary}[theorem]{Corollary}

\theoremstyle{remark}
\newtheorem{remark}{Remark}[section]
\renewcommand{\qed}{\mbox{\ \ \ }\rule{6pt}{7pt} \bigskip}
\renewcommand{\comment}[1]{}



\begin{document}
\title{How to sell an app: pay-per-play or buy-it-now?
\author{Shuchi Chawla \thanks{University of Wisconsin-Madison. \tt{shuchi@cs.wisc.edu}.}
 \and Nikhil R. Devanur\thanks{Microsoft Research. \tt{nikdev@microsoft.com}.}
 \and Anna R. Karlin \thanks{University of Washington. \tt{karlin@cs.washington.edu}.}
 \and Balasubramanian Sivan\thanks{Microsoft Research. \tt{bsivan@microsoft.com}.}
}
}

\date{}
\maketitle{}
\begin{abstract}
We consider pricing in settings where a consumer discovers his value for a good only as he uses it, and the value evolves with each use. We explore simple and natural pricing strategies for a seller in this setting, under the assumption that the seller knows the distribution from which the consumer's initial value is drawn, as well as the stochastic process that governs the evolution of the value with each use. 

We consider the differences between up-front or ``buy-it-now" pricing (BIN), and ``pay-per-play" (PPP) pricing, where the consumer is charged per use. Our results show that PPP pricing can be a very effective mechanism for price discrimination, and thereby can increase seller revenue. But it can also be advantageous to the buyers, as a way of mitigating risk. Indeed, this mitigation of risk can yield a larger pool of buyers. We also show that the practice of offering free trials is largely beneficial. 

We consider two different stochastic processes for how the buyer's value evolves: In the first, the key random variable is how long the consumer remains interested in the product. In the second process, the consumer's value evolves according to a random walk or Brownian motion with reflection at 1, and absorption at 0. 

\end{abstract}

\thispagestyle{empty}
\newpage
\setcounter{page}{1}

\section{Introduction}
\label{sec:intro}
A standard assumption in mechanism design about a consumer of goods is that he has valuations for various bundles of goods, which he knows ahead of time. The consumer is then modeled as a utility maximizer, where the utility is typically quasi-linear -- equal to the consumer's valuation minus payments made by him. While this is often a reasonable assumption, in this paper we propose an alternate model that is more appropriate for many scenarios. 
In particular, we consider scenarios where the consumer discovers his valuation of a good as he {\em uses}  it. For instance, consider a consumer buying a song on iTunes. The consumer's valuation of the song depends on how much he enjoys it as he listens to it and how many times he wants to listen to it. The consumer does not know these quantities ahead of time, and only discovers them as he repeatedly listens to the song. Indeed the consumer's enjoyment of the song, rather than being constant over time, may evolve as he listens to the song more and more. Another example is an app or a video game that a consumer plays repeatedly. As in the case of a song, the consumer's value may evolve as he uses the product repeatedly; At some point of time, the consumer may tire of the product altogether and stop using it.

We model such scenarios by having a value per usage that evolves with usage.
We use $V_t$ to denote the value to the consumer for the $t+1^{\rm th}$ usage\footnote{In some cases, it may be appropriate to model usage as a continuous quantity, such as the duration of time a game is played, but for the following discussion it is easier to think of it as discrete, and this assumption does not affect our analysis.}. The consumer knows $V_t$ only after she has used the good for $t$ times. 
Note that $t$ denotes usage and not real time: if the consumer does not use the item his valuation for the next usage does not change. 
We model the evolution of $V_t$ as a random process. The consumer has an initial value for the very first usage, $V_0$, which is sampled from a distribution. Subsequently, $V_t$ evolves according to a given random process that depends on $V_0$. 

The prevalent mechanism for selling songs, apps, software, and other digital goods is to offer a one-time contract to the consumer for unlimited usage of the product. We call such a mechanism a ``Buy-It-Now'' (\bin) scheme. Our value evolution model can be ``reduced'' back to the standard valuation model for a one-shot game by considering the expected sum of all the $V_t$s conditioned on $V_0$; This is the total expected value that the consumer obtains through unlimited use of the product \footnote{This assumes that the consumer is risk neutral, which is perhaps {\em not} a reasonable assumption in the relevant scenarios. In case of risk averse consumers, one can likewise consider the risk adjusted expectation of total value.}. However, such a reduction overlooks the fact that there are other interesting pricing mechanisms in the value evolution model. For instance, a simple pricing mechanism is to charge the consumer {\em per usage}. We call such a mechanism a ``Pay-Per-Play'' (\ppp) scheme. While any one-shot mechanism offers a uniform price for unlimited usage to all consumers, a \ppp\ scheme allows different consumer types to pay different amounts based on how long they stay interested in the product. It therefore price-discriminates in a natural and effective way. A \ppp\ scheme also gives the consumer a finer grained control over his own utility, and is therefore preferable to consumers that are risk averse.

Another example of a mechanism in the value evolution model is the common practice of offering a ``free-trial'' to the consumers. Following this free trial, consumers that stay interested may purchase the product in one shot or per-usage. Our model offers a formal explanation of the benefits of such a free trial period.

We would like to highlight the imminently practical nature of the model we consider. 
Almost none of the goods such as music, games, or apps are currently 
sold on a \ppp\ scheme, yet such a scheme would be extremely easy to implement given current technology. 
Insights obtained using our model could therefore be influential in the adoption of a new pricing scheme for such goods. 
Our initial findings suggest that different \ppp\ schemes could do better (and sometimes {\em much} better) than the standard \bin\ pricing scheme. 
Similarly, our results suggest that offering a free trial is often a very good idea. 
Yet, these findings only scratch the surface and  
they open up a whole new spectrum of open problems and research directions (which we discuss towards the end of the paper). 

We consider two different random models of value evolution, both of which are meant to be idealized versions of value evolution in some relevant markets. 
The first model is a binary value model where the consumer's value stays constant while he continues to remain interested in the product, and then drops to zero. In other words, there is a joint distribution over $V_0$ and another random variable $T$, and 
$V_t = V_0 $ for $t\leq T$ and $V_t = 0$ otherwise. 
The binary value model captures scenarios where the consumer's experience from using the product stays the same over time, but after a few uses, the consumer either does not require the product any more or switches to a competitor. Software usage is a good example of this.
The other model we study is a Martingale random walk model where 
$V_t$ is Markovian and the increments are independent with mean 0.\footnote{The model is quite general and includes as a special case the Brownian motion, a standard model of the evolution of stock prices. 
	See Appendix \ref{sec:app2} for details. } The random walk model captures the notion that the value evolves continuously with usage, each usage either increasing or decreasing the value by a small amount. 
We emphasize that the details of the models are not the focus of the paper; they are just meant to be tools for formal analysis by which we may confirm our intuition. Quantitatively, our results depend on the nature of the random process that governs value evolution, but qualitatively we draw insights that apply to both models.

We now summarize our results for the different value evolution models and risk attitudes. 
Our first (and strongest) result is for the binary value model, and the rest are for the random walk model. 
\begin{itemize}
	\item{\bf \ppp\ versus \bin\ in the \aon\ model.} We show that there exists a \ppp\ scheme whose revenue is higher than that of any \bin\ scheme even if the consumer is infinitely risk averse\footnote{See Section~\ref{sec:prelim} for a definition of infinite risk aversion.} in responding to the \ppp\ pricing, and risk neutral in responding to the \bin\ scheme. Furthermore, the \ppp\ scheme generates more consumer utility than the \bin\ scheme. Thus, the \ppp\ scheme is better than \bin\ in all respects: seller revenue, consumer utility and consumer risk. (Theorem \ref{thm:BCM}) 
	
	\item {\bf Price discrimination.} We formalize the intuition that \ppp\ schemes are good price discriminators, by showing that for a risk neutral consumer, there exists a PPP  scheme in which {\em each} consumer type, given by $V_0$, pays at least half of his {\em cumulative value} in expectation, where the cumulative value is his total value if he keeps using the good\footnote{It also turns out that the cumulative value itself is proportional to $V_0$.}. This is {\em almost perfect price discrimination}. (Perfect price discrimination would be each consumer paying his cumulative value.)  Thus, a single \ppp\ scheme obtains revenue at least $\tfrac 1 2$  of  social welfare, for any initial distribution on $V_0$. (Theorem \ref{thm:rnrn})
	
	\item {\bf Free trial.} We show that free trial can also lead to good price discrimination: appropriately choosing a trial period (number of free usages) and then offering a buy-it-now price to a risk neutral consumer also has each consumer pay a constant fraction of his cumulative value in expectation.  (Theorem \ref{thm:ftpBIN})
	
\noindent Furthermore, a combination of {\bf free trial and PPP} enables price discrimination even over infinitely risk averse consumers, and gives the {strongest result in the random walk model}: infinitely risk averse consumers pay a constant fraction of their cumulative value under this scheme. (Theorem \ref{thm:ftp})

	\item {\bf Risk aversion.} Intuitively, a \ppp\ scheme is better for a risk averse consumer, since his payment can depend on the $V_t$s and he therefore has much better control over his realized utility. At the extreme end of the spectrum, where the consumer is infinitely risk averse, we show that a \ppp\ pricing scheme can get an asymptotically\footnote{The asymptotics are with the variance going to zero, which is due to the choice that we bound the value per usage to be within  $[0,1]$. Since the quantities are scale invariant, we ``scale up'' by decreasing the variance.} unbounded factor higher revenue than \bin. 
We  also give an interpolation of this result for a single parameter family of risk averse utilities that range from  the infinitely risk averse  to  risk neutral.  (Theorems \ref{thm:rara} and \ref{thm:alphaRA})

	\item The fact that the risk neutral \ppp\ revenue is half of the social welfare also implies that it is half of the \bin\ revenue, but which is higher? The answer actually depends on the initial distribution. Given any distribution, we show how to compute the two revenues so they can be compared.  
                  While we don't have a concise characterization of distributions for which one of the schemes is better, our computations for a few distributions confirm the following intuition: distributions that are not heavily concentrated towards the upper end of $[0,1]$ favor the \ppp\ pricing scheme. Many of the well known distributions have this property, and for the uniform distribution we obtain that a \ppp\ scheme gets a higher revenue than the optimal \bin\ scheme, 
		while having each consumer type pay (in expectation) less than the optimal buy-it-now price. 
		  In other words, the \ppp\ scheme grows the total pie and hence gets a bigger share of the pie.  (Remarks \ref{rem:Compare} and \ref{rem:U01} )
\end{itemize}
Together these results give us a better understanding of the strengths and weaknesses of the different pricing schemes. They are witnesses to the kind of insights we can derive using our model of value evolution. This is yet simply scratching the surface and we expect in  future many more interesting results in this model. 
\paragraph{Related Work.}
Initiated by the work of~\citet{BB84}, there have been a number of papers~\citep{Besanko85, CL00, Battaglini05, ES07, PST14, bergemann2014dynamic} that consider revenue maximizing contracts where the private information of the agents evolves over time. The main difference from our model is that in this line of work, the agents' private information evolves over time \emph{irrespective} of whether the agent consumes the good or not. In contrast,  in our model the buyer's value evolves only if he consumes the good. The similarity to our model is that the value is often modeled as evolving according to a Markov process.  We note that both models are applicable in different scenarios: pure evolution with time is probably more applicable for a gym membership whereas ours is more applicable for products such as music, sofware, or video games. To our knowledge, the evolution of consumers' values with usage has not been studied previously in mechanism design literature. 

A technically unrelated direction, falling under the broad umbrella of online mechanisms, studies revenue/welfare maximization when buyers arrive sequentially, with either adversarial value distributions, or values drawn from a known/unknown distribution \cite{LN00, BHW02, BKRW03, KL03, BH05, BDKS12}. A special kind of online mechanism is a dynamic pricing mechanism where the seller posts a price for each buyer (as opposed to more general schemes like auctions). There is a huge body of work on dynamic pricing and revenue management problems in the Operations Research literature.  See~\cite{BZ09} and references there in. The book by~\citet{VK12} compares various pricing strategies like posted-prices, auctions, and haggling in the presence/absence of competition.


\section{Model \& Preliminaries}
\label{sec:prelim}
\paragraph{Basic Setup.} We consider a single seller offering goods for repeated private consumption (e.g. digital goods), and a single buyer. We normalize the seller's cost of production to $0$. The buyer's initial private value for the good, $V_0 \in [0,1]$, is drawn from a publicly known distribution $\dist$.  The buyer's value for the good evolves with repeated consumption of the good. He knows the process by which $V_t$ evolves but only learns his actual value $V_t$ for the $(t+1)^{st}$ consumption immediately after consuming the good $t$ times.  The seller, on the other hand, only knows the underlying distributions: the initial value distribution $F$ and the process by which the $V_t$ evolves, but not its precise value.  We assume that the buyer's value remains bounded in the range $[0,1]$ and is non-zero for only a finite number of steps.  The buyer faces uncertainty as to how his value will evolve in the future.  This uncertainty exposes him to some risk, and he takes this into account while making his purchase decisions. We elaborate further on buyer behavior below.
The seller's goal is to maximize his expected revenue, which is the sum total of all the prices paid by the buyer. 

\paragraph{Buyer's value evolution.} We consider two models for how the buyer's value evolves.
\begin{itemize}
\item {\bf The \aon\ model:}
In the \aon\ model, the buyer's value remains at his initial value $V_0$ for some number of steps $T(V_0)$, and then falls to $0$ in the $(T(V_0)+1)$th step. The number of steps $T(V_0)$ is unknown to both the buyer and the seller, but it is drawn from a publicly known distribution. This model captures settings where the buyer makes a binary decision to continue or discontinue the use of the product, but as long as he continues, he has a fixed value for the product or service. For example, the buyer may derive some fixed amount of pleasure from playing a video game or going to the gym but switches to a competitor after some rounds of play. We assume that $E[T(V_0) | V_0]$ is non-decreasing in $V_0$ and is finite.

\item {\bf The \rw\ model:}
In the  \rw\ model after each consumption, the buyer's value increases or decreases by some fixed $\delta>0$, with probability $\frac{1}{2}$ each. When the value reaches $0$, it gets absorbed there and doesn't change anymore. In other words, the buyer loses interest in the product and does not want to continue purchasing it at any price. When the value hits $1$, it gets reflected, i.e., it decreases by $\delta$. Both the buyer and the seller are aware of this process of value evolution. 

The \rw\ model is a special case of a more general {\em Markov evolution} model, in which the additive change to the buyer's value after the $t$th consumption is a random variable, $\Delta_t$, that depends only on $t$.
We assume that $\Delta_t$ has mean $0$ and standard deviation $\delta>0$ (assumed to be small).

These models capture settings where each consumption can slightly enhance or diminish the buyer's overall experience of the product. Magazine subscriptions, subscriptions to episodes of a TV serial, or to a radio channel, are some examples.

\end{itemize}

\paragraph{Selling mechanisms.} Perhaps the most prevalent selling mechanism for digital goods such as songs, apps, and video games, is what we call the {\em buy-it-now} (BIN) mechanism. Here the seller charges a one-time fee for unlimited usage of the product. We compare BIN mechanisms to schemes that charge the buyer for each successive consumption of the product. We collectively call such schemes {\em pay-per-play} (PPP) mechanisms. A pay-per-play scheme is specified by an infinite vector of prices $\prices$, with $p_t$ denoting the price charged for the $t$th consumption of the product. In this paper we focus on
simple PPP schemes such as a constant price PPP scheme. 
Pay-per-play schemes found in practice, such as magazine subscriptions and gym memberships, typically charge a fixed amount per usage. 

We also analyze the effect of offering a {\em free trial}: the seller offers the product for free for some fixed number of steps, and thereafter either offers a buy-it-now price, or a constant price PPP scheme. 

\paragraph{Buyer cumulative value, social welfare, and utility.}
The cumulative value of a buyer  is the total value he can potentially accumulate before he loses interest in the product (his value drops to $0$).  Formally, consider a buyer with initial value $V_0$ whose value evolves according to one of the models above. Let $\st$ be the buyer's stopping time -- the time at which his value drops to 0 or below.  Then the {\em cumulative value} of this buyer is the random variable $\sum_{0\le t <\st} V_t.$
We denote by $\cv(v)$ the expectation of this random variable conditioned on $V_0 = v$. 

The {\em social welfare} generated by a mechanism is the total value that the buyer obtains from using the product. In a  BIN setting, the buyer's {\em social welfare} is his cumulative value if he purchases the good, and 0 otherwise.
In a PPP setting, if we denote by $X_t$ the indicator random variable which is 1 if the buyer purchases at the $t^{th}$ opportunity, then the buyer's social welfare with a starting value of $V_0$ is $\sum_{t\ge 1} V_{t-1} X_t.$
Finally, the {\em utility}  of a buyer with initial value $V_0$ in a BIN setting  with price $p$
is $\sum_{0\le t <\st} V_t - p$ if the buyer purchases, and 0 otherwise. 
In a PPP setting, the {\em future utility} of a buyer that has purchased $t$ times, and whose current value is $V_t$, facing price sequence $\prices$  is $\sum_{s\ge t} (V_{s}-p_{s+1})X_s.$

\paragraph{Buyer's purchase behavior and risk aversion.} In a buy-it-now mechanism, the buyer makes a one-time decision to purchase or not purchase the product, depending on his (random) future utility. 
In a pay-per-play mechanism, the buyer makes a purchasing decision at every step of the process, once again taking into account the effect of his decision on his (random) future utility. Because the value of the buyer as well as the price of the product evolve with consumption, once the buyer decides to discontinue buying and consuming the product, the value and price freeze, and the buyer never buys the product again. 

We assume that the buyer is a utility maximizer, but because of the uncertainty in his future value for the item, is also sensitive to losses. In the main body of the paper, we consider the two extreme ends of the spectrum with respect to risk: risk neutral buyers on the one hand, and completely (or infinitely) risk averse buyers on the other. A {\em risk neutral} buyer chooses to buy or not depending on which of these maximizes his expected utility. An {\em infinitely risk- averse}  buyer not only ensures that his future utility is maximized in expectation, but also ensures that his future utility will be non-negative with probability $1$. We use $\rbin$ (resp. $\rppp$) to denote the revenue from a BIN scheme (resp. PPP scheme) with a risk-neutral buyer.  $\rbin[\infty]$ and $\rppp[\infty]$ denote the corresponding revenues from an infinitely risk-averse buyer.

In Appendix~\ref{sec:app}, we consider a single-parameter family of risk-averse preferences that interpolates between risk neutral and infinitely risk averse preferences. 

\paragraph{Monopoly price and single-round optimal revenue.} We will need the optimal revenue achievable in a single round when a buyer's value is drawn from the publicly known distribution $\dist$. A special case of Myerson's result~\cite{M81} shows that the optimal revenue in this case is $\max_p p(1-\dist(p))$. The price that optimizes this expression is the monopoly price, denoted by $\mnp$, is given by $\phi^{-1}(0)$, where $\phi(x) = x - \frac{1-F(x)}{f(x)}$, and the optimal revenue is denoted by $\rmye$.

\section{\Aon\ model}
\label{sec:aon}
We  begin with an analysis of the \aon\ model of value evolution. Recall that in this model, the buyer has a fixed constant value for usage of the product. However, after some random number of steps, $T(V_0)$, which is unknown to both the buyer and the seller, the buyer loses interest in the product and his value drops to $0$. We call $T(V_0)$ the buyer's stopping time. Since the buyer's value doesn't change as long as he maintains interest in the product, a simple and natural selling mechanism is a constant price \ppp\ scheme. Regardless of the buyer's risk attitude, a buyer will purchase the product in the constant price scheme if and only if $V_0$ exceeds the price, and in that case, for as long as his value is non-zero. Buyers with different stopping times pay different amounts to the mechanism, which gives an effective way to price discriminate. 

We convert this intuition into a proof, showing that a constant price \ppp\ scheme with an appropriate price obtains more revenue (even with a risk averse buyer) than a \bin\ scheme (even with a risk neutral buyer). Moreover, the \ppp\ scheme  generates more social welfare because it serves more buyers. Some of this social welfare can be shared with the buyers, resulting in a higher total utility for the buyers in addition to a higher revenue for the seller. So in this setting, \ppp\ is better in all respects than \bin.

\begin{theorem}\label{thm:BCM} 
In the \aon\ model when $\Ex{T(V_0)|V_0}$ is a non-decreasing function of $V_0$ there is a constant price \ppp\ scheme that with an {\em infinitely risk averse buyer} gets at least as much revenue, as much social welfare, and as much buyer utility as the revenue optimal \bin\ scheme with a {\em risk-netural buyer}.
\end{theorem}
\begin{proof}
In the \aon\ model with a {\em risk-neutral} buyer, recall that the buyer's expected cumulative value is $\cv(V_0) = V_0 \Ex{T(V_0)|V_0}$. The optimal \bin\ price is the monopoly price for the distribution of this random variable where $V_0$ is drawn from the distribution $\dist$. We assume that $\Ex{T(V_0)|V_0}$ is a non-decreasing function of $V_0$. So  there is a threshold initial value $v^{\bin}$ such that the risk neutral buyer purchases in the \bin\ scheme if and only if his initial value is at least $v^{\bin}$. The revenue of the \bin\ scheme is then $\rbin = \cv( v^{\bin})(1-\dist(v^{\bin}))= v^{\bin} \Ex{T(v^{\bin})} (1-\dist(v^{\bin}))$.

Now  consider a constant price \ppp\ scheme. We  begin by setting the per-play price in the scheme to be equal to the threshold value $v^{\bin}$. At this price, any buyer with $V_0\ge v^{\bin}$ purchases the product until his value becomes $0$. This is the same set of buyers that purchase in the \bin\ scheme described above. Moreover, a buyer with initial value $V_0$ pays a total of $v^{\bin}\Ex{T(V_0)}$, which is at least as large as $v^{\bin} \Ex{T(v^{\bin})} $, the amount the same buyer pays in \bin. Thus, $\rppp[\infty]\ge\rbin$. Further, since this \ppp\ scheme serves the same set of buyers as the \bin\ scheme, the two generate the same total social welfare. Now, consider gradually decreasing the per-play price in the \ppp\ scheme to below $v^{\bin}$. Then, the social welfare generated by \ppp\ increases, while its revenue may or may not decrease. We can continue decreasing the price as long as the revenue of \ppp\ stays above that of  \bin, 
and at some point, both the revenue and the buyer utility (which is social welfare $-$ revenue) of \ppp\ exceed the corresponding quantities for \bin. 
\end{proof}

\section{The \rw\ model}
\label{sec:rw}
We now switch to the \rw\ model of value evolution. Recall that in
this model, the buyer's value evolves as a random walk with step size
$\delta>0$. The value stays bounded within $[0,1]$; When it reaches
$1$, it gets reflected back to stay within the range $[0,1]$. When it
reaches $0$, it gets absorbed, or remains $0$. In Appendix~\ref{sec:app2} we
consider a more general Markov evolution model and show
that many of our results extend to that setting.

\subsection{Notation and basic facts}
\label{sec:rw-prelim}
Let $\hab{u}{v}$ denote the expected time for a random
walk starting at $u$ to hit $v$, and, $\habc{u}{v}{v'}$ denote the
expected time for a random walk starting at $u$ to hit one of $v$ or
$v'$.  We use $\Prx[v]{E}$ to denote the probability of a
(random-walk-related) event $E$ when the random walk starts at
$v$. Similarly $\Ex[v]{X}$ denotes the expected value of a
(random-walk-related) random variable $X$ when the random walk starts
at $v$. 

The following lemma is standard. See, e.g., ~\cite{levin2009markov}. Note that this lemma is about a simple random walks with step size $\delta$ (i.e., there is no reflection at $1$ or absorption at $0$). 
\begin{lemma}
\label{lem:RWbasic}
Let $X_t$ be a simple random walk with increments of $\pm \delta$, starting from $X_0 = v$ with $v\in (0,1)$.  Let $w < v < u$, Then,
\begin{align}
\label{Basic1} 
& \Prx[v]{X_t\text{  reaches u before w}}  = \frac{v-w}{u-w}
\shortintertext{In particular, } 
\notag & \Prx[v]{X_t\text{  reaches 1
    before 0}} =v
\shortintertext{and,}
\label{Basic2} 
& \habc{v}{w}{u} = \frac{(v-w)(u-v)}{\delta^2}
\end{align}
\end{lemma}

The following lemma summarizes several important facts about our random walk process. In particular, it determines the hitting times for this walk taking
reflection at $1$ and absorption at $0$ into account. Some of the
analysis later in the paper is based on the fact that  buyers with high initial values are likely to
quickly reach the maximum value. This is formalized in the third part of the following lemma, which derives the expected time to reach $1$ conditioned on reaching $1$ before $0$.

\begin{lemma}
\label{cl:nice}
Suppose that $V_t$ evolves according to the \rw\ model starting at $V_0=v$. 
%
Then \begin{equation}
\label{reach1} \Prx[v]{V_t\text{  reaches 1 before 0}} = v.
\end{equation}
For any $v>u$
\begin{equation}
\label{Nice3}  h_{vu} =\Ex[v]{\text{ time to hit }u} = \frac{(v-u)(2-v-u)}{\delta^2}  \Longrightarrow \quad
T(v) = \Ex[v]{\text{ time to hit } 0} = \frac{v(2-v)}{\delta^2}.\quad
\end{equation}
Let $\tau$ be the first time $V_t$ hits 0 or 1. Then, 
\begin{equation}
\label{Nice2}  \Ex[v]{\tau | V_{\tau}=1} = \frac{1-v^2}{3\delta^2}.
\end{equation}
\end{lemma}

\begin{proof}
Equation~\eqref{reach1} is a standard fact about this kind of random walk.
  Equation~\eqref{Nice3} follows from Lemma~\ref{lem:RWbasic},
  Equation~\eqref{Basic2}, and the reflection principle: the random
  walk with reflection at $1$ and absorption at $0$ is equivalent to a
  random walk between $[0,2]$ with absorption at both $0$ and $2$. Thus, reaching $u$ in the old walk is equivalent to reaching one of $u$ and its reflection $2-u$ in the new walk.

  Equation~\eqref{Nice2} is Exercise 17.1 in ~\cite{levin2009markov}
  (proved by considering a martingale related to $V_t$, namely $M_t =
  Y_t^3 - 3tY_t $, where $Y_t = \frac{V_t}{{\delta}}$). (See also Appendix~\ref{sec:app2}.)\end{proof}

\paragraph{Expected cumulative value.} Next we present expressions for
the buyer's cumulative value. Let $\cv(v,w)$ denote the expected
cumulative value of a buyer with starting value $v$, until the value
reaches $w$ for the first time. 

\begin{lemma}
\label{cl:SWtoW}
Let $v > w$. Then, 
\begin{align*}
\cv(v,w) &= \frac{v-w}{\delta^2}\left( v(1-v) + (1-w)(1-\delta) + \delta^2\right)
\end{align*}
\end{lemma}

\begin{proof}
Using Lemmas~\ref{lem:RWbasic} and~\ref{cl:nice}, we have that the expected cumulative value
of a buyer with starting value $v$, until the value reaches $w$ for the first time is:
\begin{align*}
\cv(v,w) &= v\cdot\Ex[v]{\text{Time to hit $w$ or $1$}}+\Prx[v]{\text{Hitting $1$ before $w$}}\cdot \cv(1,w)\\
&= v\cdot\frac{(v-w)(1-v)}{\delta^2} + \frac{v-w}{1-w}\cv(1,w)
\end{align*}
Similarly:
\begin{align*}
\cv(1,w) &= 1+ \cv(1-\delta,w)\\
&= 1 + \frac{(1-\delta)(1-\delta-w)\delta}{\delta^2}+\frac{1-\delta-w}{1-w}\cv(1,w)\\
\Rightarrow \cv(1,w) &= (1-w)\left[\frac{1}{\delta}+\frac{(1-\delta)(1-\delta-w)}{\delta^2}\right]
\end{align*}
Plugging the expression for $\cv(1,w)$ into that for $\cv(v,w)$, we obtain the result. 
\end{proof}

\begin{corollary}
\label{cl:SW}
The expected cumulative value of a buyer with $V_0=v$ is:
$$\cv(v) = \cv(v,0) =\frac{v^2(1-v)}{\delta^2}+
\frac{v}{\delta^2}\left(1-\delta+\delta^2\right) \in
\frac{v}{\delta^2}[1,5/4].$$
\end{corollary}

\paragraph{Optimal risk-neutral \bin\ revenue.} Suppose that $V_0=v$. A
risk-neutral \bin\ buyer will compute his expected cumulative value with
this initial value, namely, $\cv(v)$, and accept a price of $p$ if
and only if $\cv(v) \geq p$. Thus the risk neutral \bin\ revenue is
just the Myerson optimal revenue for the distribution of $\cv(V_0)$, and the optimal price is the monopoly price for the distribution of $\cv(V_0)$. Thus we have:

\begin{lemma}
\label{cl:OptRNBIN}
The optimal risk-neutral \bin\ revenue satisfies $\frac{\rmye}{\delta^2}\le \rbin \leq \frac{5\rmye}{4\delta^2}$ where $\rmye$ is the Myerson optimal revenue in one round for the distribution $\dist$. 
\end{lemma}

\subsection{\ppp\ enables near-perfect price discrimination over risk-neutral buyers}   
\label{sec:rnrn}

In this section we will show a surprising fact about \ppp: if the buyer is risk-neutral, \ppp\ can enable near-perfect price discrimination and obtain a constant fraction of the buyer's expected cumulative value. The \ppp\ scheme achieving this bound is extremely simple: we charge a constant price of $1/2$ to the buyer for each usage of the product.

\begin{theorem}
\label{thm:rnrn}
Consider a constant price \ppp\ scheme that charges a per-usage price of $\frac{1}{2}$. The risk-neutral revenue of this scheme is at least half of the expected cumulative value of the buyer: $$\forall v\in[0,1], \,\,\, \rppp |_{V_0=v} \geq \frac{1}{2}\cv(v).$$
\end{theorem}
\begin{proof}
How does a risk-neutral PPP buyer react to a constant price \ppp\ at a price of $p$? At any point of time, this buyer computes his expected future utility conditioned on buying at this step, and buys if this expectation is non-negative. Because the \ppp\ price stays constant, the buyer's strategy is time-independent and has a threshold behavior: the buyer continues buying until his value drops below a threshold $w$. We will now compute this $w$ as a function of the \ppp\ price $p$. Let $U(v,w,p)$ denote the buyer's expected future utility in a constant price \ppp\ scheme with per-usage price $p$ when the buyer's current value is $v$ and he continues buying until his value drops below $w$. The buyer's threshold is then the smallest $w$ such that for all $v\ge w$, $U(v,w,p)$ is non-negative.

Recall that $\cv(v,w)$ denotes the expected cumulative value of a buyer with starting value $v$, until the value reaches $w$ for the first time. We can use Lemma~\ref{cl:SWtoW} to compute $U(v,w,p)$.
\begin{align}
U(v,w,p) &= \cv(v,w) - p\cdot \Ex[v]{\text{Time to hit $w$}}
= \cv(v,w) - p\cdot\frac{(v-w)(2-w-v)}{\delta^2}\nonumber\\
&= \frac{v(v-w)(1-v)}{\delta^2} + \frac{(v-w)(1-w)(1-\delta)}{\delta^2}+(v-w)- p\cdot\frac{(v-w)(2-w-v)}{\delta^2}\nonumber\\
&= \frac{v-w}{\delta^2}\left(1-\delta-w(1-\delta-p)+v(1-v)-p(2-v)+\delta^2\right)
\label{eq:Utilvwp}
\end{align}
The RHS of equation~\eqref{eq:Utilvwp} is a decreasing function of $w$. Further, as long as $v \geq w$ and $p \leq \frac{1}{2}$, the RHS is non-negative. Therefore, at a constant price of $\frac{1}{2}$, a risk-neutral buyer keeps buying until his value hits $0$. Surprisingly, even at a value of $\delta$, the buyer is willing to pay a price of $\frac{1}{2}$ because he expects to eventually make up for the lost utility. 
This means that $ \rppp |_{V_0=v} = \frac{1}{2} T(v)$. 
On the other hand, $\cv(v) \leq  T(v)$, since the per usage values are always at most 1.
The lemma follows.  
\end{proof}
\begin{remark}
\label{rem:Surprising}
It appears counter-intuitive at first glance that with a price of $\frac{1}{2}$ the buyer keeps buying until his value hits $0$. The intuition is that although the buyer may be losing money at some particular moment in time, either his value will eventually get high, and he will have large gains for a relatively long time, or his value will go to 0 relatively quickly. The resulting expected value is positive. 
\end{remark}

\begin{remark}
\label{rem:LargePrice}
Theorem~\ref{thm:rnrn} shows that posting a price of $\frac{1}{2}$ ensures that the buyer continues buying until his value hits $0$. What if the price is larger than $\frac{1}{2}$? Then the buyer's threshold to stop buying turns out to be $w=2p-1$. To see this, 
consider a buyer at value $v=w+\delta$; If his value goes one step lower 
he will stop buying. Such a buyer's utility is $U(w+\delta,w,p) = \frac{1}{\delta}\left[(1-w-\delta)(1+w-2p)+\delta(1-p)\right]$. It follows that $U(w+\delta,w,p) \geq 0$ if $w \geq 2p-1$. 
\end{remark}

\begin{remark}
	\label{rem:Compare}
	Note that Theorem \ref{thm:rnrn} gives an exact expression for the payment of a consumer with initial value $v$ 
	in the $\tfrac 1 2 $ price \ppp~scheme,  namely 
	$\frac{1}{2}T(v) = \frac{v(2-v)}{2\delta^2}$. 
	Similarly from Corollary \ref{cl:SW}, we have an exact expression for $\cv(v)$. 
	Thus, for any given distribution, we can compute the optimal \bin~revenue and the constant price \ppp~revenue with 
	a price of $\tfrac 1 2$, and compare the two to see which is higher. 
\end{remark}

\begin{remark}
	\label{rem:U01} 
Theorem~\ref{thm:rnrn} guarantees that \ppp\ obtains at least half of the revenue of \bin\ for risk neutral buyers, however, this comparison is often quite loose. In fact, for many distributions over the initial value, 
$\rppp \geq \rbin$. For instance, when $V_0$ is drawn from the $U[0,1]$ distribution, the expected revenue of $\ppp$ is $\frac{1}{2\delta^2}\int_0^1 v(2-v)\mathrm{d}v = \frac{1}{3\delta^2}$. On the other hand, $\rbin$ can be computed as follows. We have $\cv(v,0) \approxeq \frac{v^2(1-v)}{\delta^2} + \frac{v}{\delta^2}$. The BIN optimal revenue is obtained by computing $\max_t \cv(t,0)(1-\dist(t))$. Taking $F(t) = t$, the maximum value for this expression can be computed to be $\frac{5}{16\delta^2} < \frac{1}{3\delta^2}$. 
Moreover, the optimal \bin~price is $ > \tfrac 1 {2\delta^2}$ (corresponding to a threshold initial value of $\tfrac 1 2$), while the expected payment of a consumer with value $v$ in  \ppp~is always $\leq \tfrac 1 {2\delta^2} $. 
\end{remark}
\begin{remark}
On the other hand, there are distributions for which the revenue of any {\em constant price} \ppp\ scheme is strictly smaller than $\rbin$. For instance when $F(x) = x^2$, a constant price \ppp\ with a price of $1/2$ obtains revenue $\frac{5}{12\delta^2}$; Using the discussion in remark~\ref{rem:LargePrice}, it is straightforward to show that prices larger or smaller than $\frac{1}{2}$ perform no better. On the other hand, we can compute the revenue of the optimal \bin\ scheme as $\rbin = \max_t \cv(t,0)(1-F(t)) = \frac{0.479}{\delta^2}$. 
\end{remark}

\subsection{Free trial enables near-perfect price
  discrimination over all risk profiles}

\label{sec:ftp}

\newcommand{\ftp}{\ensuremath{\mathrm{free}}}
\newcommand{\rftp}[1][0]{\mathcal{R}^{\scriptscriptstyle \ftp+\ppp}_{#1}}
\newcommand{\ftb}{\ftp+\bin}
\newcommand{\rftb}[1][0]{\mathcal{R}^{\scriptscriptstyle \ftb}_{#1}}

We now show that near-perfect price discrimination can be achieved
even over risk averse buyers, if the pricing scheme allows the buyer
to try the product for free for a carefully chosen number of
trials. We first discuss a \ppp\ scheme with a free trial period that achieves this. Next
we show that even \bin\ with a free trial period enables near-perfect
price discrimination over risk neutral agents.

In the following lemmas, let  $\rftp[](T, c)$ be the expected revenue
of the seller obtained by offering the buyer a free trial period
of length $T$, and then offering him a pay-per-play price of
$c$. Likewise, let  $\rftb[](T, c)$ be the expected revenue
of the seller obtained by offering the buyer a {\em free trial period}
of length $T$, and then offering him a buy-it-now price of
$c$.

\begin{theorem}
\label{thm:ftp}
Let $T=\frac{2}{3\delta^2}$ and $c=0.089$. Then the expected seller
revenue from a \ppp\ pricing scheme with a free trial priod of length
$T$ and \ppp\ price $c$ for a risk averse buyer is a constant fraction
of the buyer's expected cumulative value:
$$\forall v\in [0,1], \,\,\, \rftp[\infty](T,c) |_{V_0=v} = \Theta(\cv(v)).$$
\end{theorem}

\begin{proof}
By Equation~\eqref{Nice2} in Lemma~\ref{cl:nice}, and Markov's inequality with $T=
\frac{2}{3\delta^2}$,
\begin{equation}
\label{eqn:TimeHit1}
\Prx[v]{\tau > T \Big| V_{\tau }= 1} \le \frac{1}{2}.
\end{equation}
Let $T'$ be the length of time after the free trial period ends for
which $V_t$ is at least $c$: $T'=\min\{t-T: t>T, \text{and, } V_{t+1}<c\}$. Then,
$\rftp[\infty] = \Ex{cT'}$. Let $\tau$ denote the time at which the
buyer's value reaches $0$ or $1$. Then for a buyer with starting value
$V_0=v$ we obtain
\begin{align}
\notag \rftp[\infty] (T,c) &\ge \Ex[v]{cT' ~| ~V_{\tau}=1\text{ and } \tau \le T}\Prx[v]{\tau \le T ~ | ~V_{\tau}=1}\Prx[v]{V_{\tau}=1}\\
\tag{by Inequality (\ref{eqn:TimeHit1}) and
 Equation~\eqref{Basic1} in Lemma \ref{lem:RWbasic}}
&\ge c \Ex{h_{1c} - T } \frac{1}{2} v\\
\tag{by Equation~\eqref{Nice3} in Lemma~\ref{cl:nice}}
& \ge c \cdot\frac{(1-c)^2 - 2/3 } {\delta^2} \cdot \frac{v}{2}
\end{align}
This latter quantity for sufficiently small $c$ is $\Omega
\left(\frac{v}{\delta^2}\right)$. It is maximized at $c\approx
0.089$. Corollary~\ref{cl:SW} then implies the lemma.
\end{proof}

\begin{theorem}
\label{thm:ftpBIN}
Let $T=\frac{3}{8\delta^2}$ and $c=\frac{1}{4\delta^2}$. Then the
expected seller revenue from offering a free trial priod of length $T$
and then a \bin\ price of $c$ to a risk neutral buyer is a constant
fraction of the buyer's expected cumulative value:
$$\forall v\in [0,1], \,\,\, \rftb(T,c) |_{V_0=v} = \Theta(\cv(v)).$$
\end{theorem}
\begin{proof}
 Suppose that the value of the risk neutral buyer at the end of the
  trial period, $V_T$ is at least $c\delta^2$. Then, noting that
  $\cv(x) \geq \frac{x}{\delta^2}$, the buyer accepts a \bin\
  price of $c$. We now lower bound the probability $Q(v)$ that $V_T
  \geq c\delta^2$ given $V_0 = v$. We can write $Q(v) \geq Q_1(v)
  Q_2$, where
$$Q_1(v) = \Prx{V_{\tau} = 1\text{ and } \tau \leq T | V_0 = v}\quad \text{ and }\quad Q_2 = \Prx{V_T \geq c\delta^2 | V_{\tau} = 1 \text{ and } \tau \leq T}.$$

\noindent By Equation~\eqref{Basic1} in Lemma~\ref{lem:RWbasic}, 
Equation~\eqref{Nice2} in Lemma \ref{cl:nice}, and Markov's inequality, with $T=\frac{3}{8\delta^2}$, we have $Q_1(v) \geq \frac{v}{9}$. 

Next, recall that if $X = \sum_{i=1}^{n} X_i$ with each $X_i$ being $\pm 1$
with probability half each, independently, then $\Prx{|X| \geq a} \leq
e^{-\frac{a^2}{2n}}$. We can bound $Q_2$ from below by setting $n=T-\tau$,
and $X_i$ to be $1/\delta$ times the change in the buyer's value at
step $\tau+i$, and requiring that the net change in value from step
$\tau$ to step $T$ is no more than $(1-c\delta^2)$. Thus,
\begin{align*}
Q_2 &\geq 1- \Prx{|X| \geq \frac{1-c\delta^2}{\delta}}\geq 1- 2e^{-\frac{(1-c\delta^2)^2}{2(T-\tau)\delta^2}}\geq 1 - 2e^{-3/4} = 0.055
\end{align*} 

Clearly, $\rftb \geq Q(v)\cdot c \geq \frac{0.055\cdot v}{9}\cdot\frac{1}{4\delta^2} = \Theta(\cv(v))$.
\end{proof}

\begin{remark}
An interesting fact about Theorems~\ref{thm:ftp} and~\ref{thm:ftpBIN} is that the proposed pricing
does not depend on the distribution of the buyer's initial value. This
demonstrates that in the \rw\ model the evolution of the value is a
more important phenomenon than the initial value at which the walk
starts. 
\end{remark}

\subsection{Comparing \bin\ and \ppp}
\label{sec:rnra}

Theorems~\ref{thm:rnrn} and \ref{thm:ftp} compare the \ppp\ revenue to the buyer's cumulative value under various risk profiles. As a direct consequence, \ppp\ gets a large fraction of the revenue of any \bin\ scheme as well. In this section, we make the comparison between \ppp\ and \bin\ more precise and show that in some cases \ppp\ beats \bin\ by far. The following theorem deals with the case of risk averse buyers.
(In Appendix~\ref{sec:app} we extend this result to the case of $\alpha$-risk averse buyers for some $\alpha\in [0,\infty)$ achieving a tradeoff between \ppp\ and \bin\ revenue as a function of the risk parameter.)
\begin{theorem}
\label{thm:rara}
With an infinitely risk averse buyer, there exists a constant price \ppp\ scheme with revenue $\rppp[\infty] \geq \frac{1}{4\delta}\rbin[\infty]$. 
\end{theorem}
\begin{proof}
{\bf Revenue of \bin.} In \bin, the infinitely risk-averse buyer will first compute his cumulative value under the worst possible evolution of values, namely, when his value goes down by $\delta$ in every single round, and thus, will hit $0$ in $v/\delta$ rounds. The buyer's worst possible cumulative value is then $\frac{v^2}{2\delta}$. From this it follows that $\rbin[\infty] = \max_t \frac{t^2}{2\delta}(1-F(t))$.

\noindent
{\bf Revenue of \ppp.} Let $t^* = \argmax_t \frac{t^2}{2\delta}(1-F(t))$. Consider a constant price \ppp\ scheme with a price of $\frac{t^*}{2}$ per round. We get the revenue from an infinitely risk averse buyer to be:
\begin{align*}
\rppp[\infty] &= \frac{t^*}{2} \text{ times the expectation over $V_0$ of the time to hit } t^*/2\\
&= \frac{t^*}{2}\frac{1}{\delta^2}\int_{\frac{t^*}{2}}^1\left(v-\frac{t^*}{2}\right)\left(2-v-\frac{t^*}{2}\right)\dens(v)\mathrm{d}v\\
&\geq \frac{(t^*)^2}{4\delta^2}\int_{t^*}^1\left(2-v-\frac{t^*}{2}\right)\dens(v)\mathrm{d}v\\
&\geq \frac{(t^*)^2}{8\delta^2}\int_{t^*}^1\dens(v)\mathrm{d}v = \frac{(t^*)^2}{8\delta^2}(1-F(t^*))= \frac{1}{4\delta}\rbin[\infty]
\end{align*}
Here the first inequality follows by restricting the limit of the integral to values greater than $t^*$ and the second follows by noting $v,t^*\le 1$.
\end{proof}

\begin{remark}
The factor $\Theta(\frac{1}{\delta})$ is tight. Consider a point mass distribution at $v$. Here, we have $\rbin[\infty] = \frac{v^2}{2\delta}$, while $\rppp[\infty] = \max_p p\cdot \frac{(v-p)(2-p-v)}{\delta^2} \leq \frac{v^2}{\delta^2}$. 
\end{remark}

Next we show that in some settings, \ppp\ with a constant price obtains revenue from a risk averse buyer that compares well even to the \bin\ revenue from a risk neutral buyer.
\begin{theorem}
\label{thm:rnra}
There exists a \ppp\ scheme, that, with an {\em infinitely risk-averse buyer} gets a revenue at least a $\frac{\mu}{10}$ fraction of the revenue optimal BIN scheme with a {\em risk-neutral buyer}, where $\mu$ is the 
monopoly price of the distribution $\dist$ of the initial value $V_0$.
\end{theorem}
\begin{proof}

 Recall that by Lemma \ref{cl:OptRNBIN}, the optimal risk-neutral \bin\ revenue satisfies $\rbin \leq \frac{5\rmye}{4\delta^2}$.

We now consider the infinitely risk-averse \ppp\ revenue for a constant per-play price $p$. This revenue is given by $\rppp[\infty] = p\cdot\int_p^1 h_{vp}f(v)\mathrm{d}v$, where $h_{vp}$, as before, is the time taken by the random walk starting at $v$ to go strictly below $p$. By Equation~\eqref{Nice3} in Lemma~\ref{cl:nice}, we have $h_{vp} = \frac{(v-p)(2-p-v)}{\delta^2}$. Now consider a price of $p =\frac{\mu}{2}$. Since $p<1/2$, we have $2-p-v \geq \frac{1}{2}$. Thus, as in the proof of the previous lemma, we have,
\begin{align*}
  \rppp[\infty] & \geq \frac{p}{2\delta^2}\int_p^1 (v-p)f(v)\mathrm{d}v
\ge \frac{\mu}{4\delta^2}\frac{\mu}{2}\int_{\mu}^1 f(v)\mathrm{d}v 
 = \frac{\mu}{8\delta^2}\rmye \geq \frac{\mu}{10}\rbin.
\end{align*}
\end{proof}

\begin{remark}
When the monopoly price $\mu$ of the initial value distribution $\dist$ is a constant, \ppp\ with a risk-averse buyer gets a constant fraction of \bin\ with a risk-neutral buyer. Otherwise, \ppp\ gets a very small fraction of the \bin\ revenue. This is not just an artifact of our analysis: suppose that the buyer's initial value is $v$ with probability $1$. Then \bin\ gets a revenue of $\Omega(v/\delta^2)$. On the other hand, \ppp\ with a constant price of $p$ can obtain at most $2p(v-p)/\delta^2$, which is a factor of $\Omega(v)$ smaller than the $\rbin$.
\end{remark}

\section{Conclusions}
\label{sec:con}
In this paper, we have studied  pricing in settings where a consumer learns his value for a good only as he uses it. We have seen that in such a setting, the use of ``pay-per-play" pricing can be advantageous to both the seller, in terms of revenue, and to the buyer, as a mechanism for reducing risk. Likewise, offering a free trial period can improve the seller's revenue.

The results in this paper are just first steps. Some natural future questions include the following: (1) What is the optimal seller mechanism and how does it depend on how risk-averse the buyer is? As an intermediate model between \ppp\ and \bin\ pricing, it may help to price usage at regular intervals, as a sort of ``subscription'' pricing. (2) Within the framework of \ppp\ pricing, is it a good idea to set a constant price per usage, or does it help to increase or decrease price over time? (3) What about other models for value evolution? For example, it would be interesting to incorporate drift, either positive or negative, into the buyer's value evolution. Likewise, can we say something interesting when per-step changes in value are not always small? (4) Can we generalize these results to incorporate other models of risk? (5) What happens when multiple schemes co-exist, e.g. when the consumer is offered both a \bin~price and a \ppp~scheme? 

It would be very interesting to also have empirical estimation of consumer behavior to complement the theoretical results. 
Experiments that shed light on what is a good model of value evolution would be extremely useful. 
Similarly, risk aversion models that explain consumer behavior in these markets would add great value.

\bibliographystyle{plainnat}
\bibliography{bib_file}

\begin{thebibliography}{17}
\providecommand{\natexlab}[1]{#1}
\providecommand{\url}[1]{\texttt{#1}}
\expandafter\ifx\csname urlstyle\endcsname\relax
  \providecommand{\doi}[1]{doi: #1}\else
  \providecommand{\doi}{doi: \begingroup \urlstyle{rm}\Url}\fi

\bibitem[Babaioff et~al.(2012)Babaioff, Dughmi, Kleinberg, and
  Slivkins]{BDKS12}
Moshe Babaioff, Shaddin Dughmi, Robert Kleinberg, and Aleksandrs Slivkins.
\newblock Dynamic pricing with limited supply.
\newblock In \emph{{ACM} Conference on Electronic Commerce, {EC} '12, Valencia,
  Spain, June 4-8, 2012}, pages 74--91, 2012.

\bibitem[Bar-Yossef et~al.(2002)Bar-Yossef, Hildrum, and Wu]{BHW02}
Ziv Bar-Yossef, Kirsten Hildrum, and Felix Wu.
\newblock Incentive-compatible online auctions for digital goods.
\newblock In \emph{Proceedings of the Thirteenth Annual ACM-SIAM Symposium on
  Discrete Algorithms}, SODA '02, pages 964--970, 2002.

\bibitem[Baron and Besanko(1984)]{BB84}
David~P. Baron and David Besanko.
\newblock Regulation and information in a continuing relationship.
\newblock \emph{Information Economics and Policy}, 1\penalty0 (3):\penalty0 267
  -- 302, 1984.

\bibitem[Battaglini(2005)]{Battaglini05}
Marco Battaglini.
\newblock Long-term contracting with markovian consumers.
\newblock \emph{American Economic Review}, 95\penalty0 (3):\penalty0 637--658,
  2005.

\bibitem[Bergemann and Strack(2014)]{bergemann2014dynamic}
Dirk Bergemann and Philipp Strack.
\newblock Dynamic revenue maximization: A continuous time approach.
\newblock 2014.

\bibitem[Besanko(1985)]{Besanko85}
David Besanko.
\newblock Multi-period contracts between principal and agent with adverse
  selection.
\newblock \emph{Economics Letters}, 17\penalty0 (1–2):\penalty0 33 -- 37,
  1985.

\bibitem[Besbes and Zeevi(2009)]{BZ09}
Omar Besbes and Assaf Zeevi.
\newblock Dynamic pricing without knowing the demand function: Risk bounds and
  near-optimal algorithms.
\newblock \emph{Operations Research}, 57\penalty0 (6):\penalty0 1407--1420,
  November 2009.

\bibitem[Blum and Hartline(2005)]{BH05}
Avrim Blum and Jason~D. Hartline.
\newblock Near-optimal online auctions.
\newblock In \emph{Proceedings of the Sixteenth Annual {ACM-SIAM} Symposium on
  Discrete Algorithms, {SODA}}, pages 1156--1163, 2005.

\bibitem[Blum et~al.(2003)Blum, Kumar, Rudra, and Wu]{BKRW03}
Avrim Blum, Vijay Kumar, Atri Rudra, and Felix Wu.
\newblock Online learning in online auctions.
\newblock In \emph{Proceedings of the Fourteenth Annual {ACM-SIAM} Symposium on
  Discrete Algorithms, January 12-14, 2003, Baltimore, Maryland, {USA.}}, pages
  202--204, 2003.

\bibitem[Courty and Hao(2000)]{CL00}
Pascal Courty and Li~Hao.
\newblock Sequential screening.
\newblock \emph{Review of Economic Studies}, 67\penalty0 (4):\penalty0
  697--717, 2000.

\bibitem[Eso and Szentes(2007)]{ES07}
Peter Eso and Balázs Szentes.
\newblock Optimal information disclosure in auctions and the handicap auction.
\newblock \emph{Review of Economic Studies}, 74\penalty0 (3):\penalty0
  705--731, 2007.

\bibitem[Kleinberg and Leighton(2003)]{KL03}
Robert~D. Kleinberg and Frank~Thomson Leighton.
\newblock The value of knowing a demand curve: Bounds on regret for online
  posted-price auctions.
\newblock In \emph{44th Symposium on Foundations of Computer Science {(FOCS}},
  pages 594--605, 2003.

\bibitem[Lavi and Nisan(2000)]{LN00}
Ron Lavi and Noam Nisan.
\newblock Competitive analysis of incentive compatible on-line auctions.
\newblock In \emph{Proceedings of the 2Nd ACM Conference on Electronic
  Commerce}, EC '00, pages 233--241, 2000.

\bibitem[Levin et~al.(2009)Levin, Peres, and Wilmer]{levin2009markov}
David~Asher Levin, Yuval Peres, and Elizabeth~Lee Wilmer.
\newblock \emph{Markov chains and mixing times}.
\newblock American Mathematical Society, 2009.

\bibitem[Myerson(1981)]{M81}
R.~Myerson.
\newblock Optimal auction design.
\newblock \emph{Mathematics of Operations Research}, 6:\penalty0 58--73, 1981.

\bibitem[Pavan et~al.(2014)Pavan, Segal, and Toikka]{PST14}
Alessandro Pavan, Ilya Segal, and Juuso Toikka.
\newblock Dynamic mechanism design: A myersonian approach.
\newblock \emph{Econometrica}, 82\penalty0 (2):\penalty0 601--653, 2014.

\bibitem[Vohra and Krishnamurthi(2012)]{VK12}
Rakesh~V. Vohra and Lakshman Krishnamurthi.
\newblock \emph{Principles of Pricing - An Analytical Approach}.
\newblock Cambridge University Press, 2012.
\newblock ISBN 978-1-107-01065-9.

\end{thebibliography}
\appendix

\section{$\alpha$-risk averse BIN vs $\alpha$-risk averse PPP}
\label{sec:app}
We have shown that in the two extreme cases, namely, with risk neutral buyers (Section~\ref{sec:rnrn}) and with infinitely risk averse buyers (Section~\ref{sec:rnra}) PPP does well compared to BIN. Is it possible to compare the revenue of BIN and the revenue of PPP as a function of how risk-averse the buyer is? We address this question in this section. 

\paragraph{$\alpha$-risk averse buyer.}
We parameterize risk by the maximum loss the buyer is willing to suffer.  Recall that a completely risk averse buyer not only ensures that his future utility is maximized in expectation, but also ensures that his future utility will be non-negative with probability $1$. On the other hand, an $\alpha$-risk averse buyer, for some $\alpha>0$, is willing to tolerate a loss of up to $1/\alpha$ while maximizing his expected future utility. 

Formally, we have:
\begin{itemize}
\item {\bf $\alpha$-risk averse BIN buyer.} Given a BIN price of $p$, for every $\alpha \geq 0$, an $\alpha$-risk averse buyer buys if and only if his future utility is non-negative in expectation and no smaller that $-1/\alpha$ with probability $1$. Formally, $\cv(V_0) - p \geq 0$, and, $\Prx{\sum_{0\le t <\st} V_t - p < -\frac{1}{\alpha}} = 0$.
\item {\bf $\alpha$-risk averse PPP buyer.} The behavior of the buyer in a PPP mechanism is more complex. As we mentioned earlier, the buyer buys the product for some number of steps and then stops buying it. After the $t$th step, if the buyer's value $V_t$ is larger than the current price $\price_t$ in the PPP scheme, then it is always in the buyer's best interest to continue buying at that step. If $V_t$ is less than $\price_t$, however, then the buyer may still decide to continue buying, banking on the possibility of obtaining a large positive utility later on. Fixing the vector of PPP prices $\prices$, the buyer's strategy specifies as a function of his current value whether or not he should continue to purchase the product. We can define the optimal such strategy via backwards induction.  In essence, the optimal PPP buyer strategy maximizes his expected utility, subject to the constraint that his expected utility never falls below $\frac{-1}{\alpha}$. We omit the formal specification of the optimal strategy here because for the price vectors we construct, the optimal strategy is immediate.
\end{itemize}

\begin{remark}
An $\alpha$-risk averse buyer with $\alpha = 0$ corresponds to a risk neutral buyer, and at $\alpha = \infty$ corresponds to an infinitely risk averse buyer. 
\end{remark}

We now construct a PPP scheme that obtains a revenue comparable to that of the optimal \bin\ scheme under any risk profile, that is, any $\alpha> 0$. Our \ppp\ scheme charges a fixed price of $p$ per usage for the first few rounds, and then charges a price of $0$ for the remaining rounds. We call such a scheme a {\em Rent-To-Own} scheme: after the buyer has paid some minimum amount of money to the seller, he gets unlimited use of the product for free. Recall that $\rppp[\alpha]$ and $\rbin[\alpha]$ denote the revenues of the optimal \ppp\ and \bin\ schemes respectively, with an $\alpha$-risk averse buyer.
\begin{theorem}\label{thm:alphaRA} 
Given an $\alpha$-risk averse buyer with $\alpha \in (0,\infty)$, let $v^*$ be the smallest buyer value that accepts the optimal $\bin$ price. Then,
\begin{enumerate}
\item $\rbin[\alpha] \leq \left[\frac{(v^*)^2}{2\delta} + \min(\frac{1}{\alpha},\cv(v^*))\right](1-\dist(v^*))$
\item There exists a PPP scheme with revenue $\rppp[\alpha]  \geq \Theta(1)\cdot\left[\frac{(v^*)^2}{\delta^2} + \min(\frac{1}{\alpha},\cv(v^*))\right](1-\dist(v^*))$ 
\end{enumerate}
In particular, for all $\alpha$, $\rppp[\alpha] \geq \frac{1}{32}\rbin[\alpha]$, and when $\frac{1}{\alpha} = O(\frac{(v^*)^2}{2\delta})$, we have $\rppp[\alpha] = \omega(\rbin[\alpha])$. 
\end{theorem}
\begin{proof}
{\bf BIN revenue.} If a buyer with value $v^*$ accepts the optimal $\bin$ price, it follows that the optimal $\bin$ price $p^* \leq\cv(v^*)$. Also, by the definition of an $\alpha$-risk averse buyer, the optimal $\bin$ price $p^*$ is at most $\frac{1}{\alpha}$ plus the cumulative value obtained by $v^*$ when the value evolves by decreasing in every round. The latter is the area of a triangle with height $v^*$ and base $\frac{v^*}{\delta}$. Thus, the optimal $\bin$ price $p^* \leq \frac{1}{\alpha} + \frac{(v^*)^2}{2\delta}$. Since we have shown two upper bounds on $p^*$, it follows that $p^* \leq \min\left(\frac{1}{\alpha}+\frac{(v^*)^2}{2\delta}, \cv(v^*)\right) \leq \frac{(v^*)^2}{2\delta} + \min(\frac{1}{\alpha},\cv(v^*))$. Since only buyers with value no smaller than $v^*$ buy, $\rbin[\alpha] \leq \left[\frac{(v^*)^2}{2\delta} + \min(\frac{1}{\alpha},\cv(v^*))\right](1-\dist(v^*))$.

{\bf PPP revenue.} Suppose that among the two terms in the upper bound of $p^*$, namely 
$\frac{(v^*)^2}{2\delta}$ and $\min(\frac{1}{\alpha},\cv(v^*))$, we have $\frac{(v^*)^2}{2\delta} \geq \sqrt{\delta}\cdot\min(\frac{1}{\alpha},\cv(v^*)) $. In this case the proof of Theorem~\ref{thm:rara} shows that there exists a constant price PPP with revenue $\frac{(v^*)^2}{8\delta^2}(1-\dist(v^*)) \geq \Omega(\frac{1}{\sqrt{\delta}})\rbin[\alpha]$, proving the theorem, and showing that $\rppp[\alpha] = \omega(\rbin[\alpha])$ in this case. 

Supoose on the other hand, $\frac{(v^*)^2}{2\delta} < \sqrt{\delta}\min(\frac{1}{\alpha},\cv(v^*))$. Consider a dual price PPP scheme that charges a price of $\min\big(\frac{1}{2}, \frac{1}{c\alpha\cv(v^*)}\big)$ for $c\cv(v^*)$ rounds, and $0$ thereafter, where $c = 24$. Note that with this PPP scheme, the total possible loss for the buyer is upper bounded by $\frac{1}{c\alpha\cv(v^*)}\cdot c\cv(v^*) \leq \frac{1}{\alpha}$. 
From the proof of Theorem~\ref{thm:rnrn} we know that given a price no larger than $\frac{1}{2}$, every buyer type except $0$ gets a non-negative future utility by continually buying till the end, where as the buyer is sure to get a $0$ utility by rejecting in this round. Since the current round price of $\min\left(\frac{1}{2},\frac{1}{c\alpha\Ex{CV(v^*)}}\right) \leq \frac{1}{2}$, it follows that every buyer type except $0$ will buy till the end, resulting in revenue of 
$$\rppp[\alpha] \geq \min\left(\frac{1}{2},\frac{1}{c\alpha\cv(v^*)}\right)\cdot (1-\dist(v^*))\cdot \Ex[v^*]{\min\left(h_{v^*,0}, c\cv(v^*)\right)}.$$

The term $\Ex[v^*]{\min\left(h_{v^*,0}, c\cv(v^*)\right)}$ gives a lower bound on the number of rounds that the seller can collect a price of $\min\left(\frac{1}{2},\frac{1}{c\alpha\cv(v^*)}\right)$. For the buyer with value $v^*$, the number of rounds where the seller can collect this price is exactly $\Ex[v^*]{\min\left(h_{v^*,0}, c\cv(v^*)\right)}$, and the number of rounds is larger for larger $v^*$. Given that $\Ex{h_{v^*,0}} = \frac{v^*(2-v^*)}{\delta^2} \leq \frac{2v^*}{\delta^2}$, and $\cv(v^*) \geq \frac{v^*}{\delta^2}$, we have $\cv(v^*) \geq \frac{\Ex{h_{v^*,0}}}{2}$. From Lemma~\ref{lem:hittingTime} we know that $\Prx{h_{v^*,0} \geq k\cdot\Ex{h_{v^*,0}}} \leq 2^{-\lfloor\frac{k}{2}\rfloor}$. Setting $c = 24$, we get
\begin{align*}
\Ex[v^*]{\min\left(h_{v^*,0}, 24\cv(v^*)\right)} &\geq \Ex[v^*]{\min\left(h_{v^*,0}, 12\Ex{h_{v^*,0}}\right)}\\
&\geq \Ex{h_{v^*,0}} - \sum_{\ell=6}^{\infty}(2^{-\ell}-2^{-(\ell+1)})\cdot (2\ell+2) \cdot \Ex{h_{v^*,0}}\\
&= \Ex{h_{v^*,0}}\cdot\left(1-0.5\cdot\sum_{\ell=6}^{\infty}(2\ell+2)\cdot 2^{-\ell}\right)\\
&= \frac{3}{4}\Ex{h_{v^*,0}}\\
&\geq \frac{3}{4}\cv(v^*)
\end{align*}

If $\min\left(\frac{1}{2},\frac{1}{24\alpha\cv(v^*)}\right) = \frac{1}{2}$, then 
$$\rppp[\alpha] \geq \frac{1}{2}\cdot \frac{3}{4}\cdot\cv(v^*)(1-\dist(v^*)) = \frac{3}{8}\cv(v^*)(1-\dist(v^*)) \geq \frac{3}{8}\rbin[\alpha],$$
where the last inequality follows from $\rbin[\alpha] \leq \cv(v^*)(1-\dist(v^*))$. 

If $\min\left(\frac{1}{2},\frac{1}{24\alpha\cv(v^*)}\right) = \frac{1}{24\alpha\cv(v^*)}$, then
$$\rppp[\alpha] \geq \frac{1}{24\alpha\cv(v^*)}\cdot \frac{3}{4}\cdot\cv(v^*)(1-\dist(v^*)) \geq \frac{1}{32\alpha}(1-\dist(v^*)) \geq \frac{1}{32(1+\sqrt{\delta})}\rbin[\alpha],$$
where the last inequality follows from $\rbin[\alpha] \leq \left(\frac{1}{\alpha} + \frac{(v^*)^2}{2\delta}\right)(1-\dist(v^*)) \leq (1+\sqrt{\delta})\frac{1}{\alpha}(1-\dist(v^*))$, given our assumption that $\frac{(v^*)^2}{2\delta} < \sqrt{\delta}\min(\frac{1}{\alpha},\cv(v^*))$. As $\delta \to 0$, this is a $32$ approximation, thus proving the theorem. 
\end{proof}
\begin{lemma}
\label{lem:hittingTime}
The time $h_{v,0}$ to hit $0$ starting from $v$ satisfies $\Prx{h_{v,0} \geq k\cdot\Ex{h_{v,0}}} \leq 2^{-\lfloor\frac{k}{2}\rfloor}$ for $k \geq 2$.
\end{lemma}
\begin{proof}
From Markov's inequality we have $\Prx{h_{v,0} \geq 2\cdot\Ex{h_{v,0}}} \leq \frac{1}{2}$. By reflection principle, a random walk with reflection at $1$ hitting $0$ is the same as a random walk without reflection hitting one of $0$ or $2$. Thus, the random walk with reflection at $1$ not having hit $0$ after $2\Ex{h_{v,0}}$ steps is equivalent to the random walk without reflection not having hit either of $0$ or $2$ after $2\Ex{h_{v,0}}$ steps.  
The position of the random walk is now a random variable $X$ with expectation $v$. Our goal is to repeatedly apply Markov's inequality. To do this, we need an upper bound on $\Ex{h_{X,0}}$.  We have $\Ex{h_{X,0}} = \frac{2\Ex{X}-\Ex{X^2}}{\delta^2} \leq \frac{2\Ex{X}-\Ex{X}^2}{\delta^2}$. Substituting $\Ex{X} = v$, we get $\Ex{h_{X,0}} \leq \frac{2v-v^2}{\delta^2} = \Ex{h_{v,0}}$. Thus, the probability that, after $2\Ex{h_{v,0}}$ more steps the random walk did not hit $0$ or $2$ incurs an additional factor of at most $\frac{1}{2}$ by Markov's inequality. Applying Markov's inequality $\lfloor\frac{k}{2}\rfloor$ times, we get the lemma.
\end{proof}

\newcommand{\tDelta}{\tilde{\Delta}}

\section{A more general random walk model}
\label{sec:app2}

In Section~\ref{sec:rw-prelim}, we assumed that the value of the buyer evolved according to a simple random walk with a step size of $\delta$, with reflection at 1 and absorption at 0.  The same results go through assuming the value evolves according to Brownian motion with mean 0 and standard deviation $\delta$.  In particular, the necessary lemmas used in the proofs, e.g., all the results in Section~\ref{sec:rw-prelim}, hold essentially unchanged. (See the book by Levin, Peres and Wilmer~\cite{levin2009markov}, Theorem 2.49, Page 57 and Exercise 2.17 (b), page 62.)

In this section, we consider a more general discrete time random walk
model, and outline some of the (standard) arguments used to prove the
claims in Section~\ref{sec:rw-prelim}. Specifically, we assume that
the value $V_t$ evolves according to a discrete time Markov process,
which is a martingale (except at the boundaries). We can also forego
the assumption that the walk is symmetric (implying the third moment
of the steps is 0).  We will assume though that: (a) $V_0 \in (0,1)$.
(b) Reflection occurs whenever $V_t$ crosses 1 (from
below). Specifically, if $V_{t} > 1$, then $V_{t+1} := V_{t-1}$. Note
that in a departure from the basic random walk model, we will allow
the value to temporarily exceed $1$ to simplify our analysis. (c)
There is absorption (i.e., the buyer loses interest forever), whenever
$V_t$ drops to 0 or below. (d) There is a constant $\epsilon>0$ which
is an upper bound on $\Delta_t := |V_{t+1} - V_t|$ for every $t$. In
what follows we assume that $\epsilon$ is sufficiently small.

Finally, we use the notation: 
$$\delta^2 := \Ex{(\Delta_t )^2| V_t}\quad\text{ and }\quad c_3 := \Ex{(\Delta_t )^3| V_t}.$$

\begin{definition}
\label{defn:tau}
Define the stopping time $\tau$ to be the first time at which $V_t \ge 1$ or $V_t \le 0$. 
\end{definition}
Our theorems in Section~\ref{sec:rw} follow from the following facts that are generalizations
of Lemmas~\ref{lem:RWbasic} and \ref{cl:nice}, and Corollary~\ref{cl:SW}.

\begin{enumerate}
\item \label{Nice1'} $\Prx[v]{V_t\text{  hits 1 before 0}}\in v \pm
  O(\epsilon)$. \hfill (Equation~\eqref{reach1} in Lemma~\ref{cl:nice})
\item \label{Nice2'} Let $\tau$ be the stopping time as defined
  above.Then \hfill (special case of Equation~\eqref{Basic2} in Lemma~\ref{lem:RWbasic})
$$\Ex[v]{\tau} \in \frac{v(1-v) \pm
    O(\epsilon)}{\delta^2}.$$ 

\item \label{Nice3'} Let $w= V_{\tau +1}$. Then for any $x< w$, \hfill (special case of Equation~\eqref{Nice3} in Lemma~\ref{cl:nice})
$$h_{w,x} |_{V_{\tau} \ge 1} =\Ex[w]{\text{ time to hit }x\text{ or
    less} | V_{\tau}\ge 1} = \frac{(1-x)^2\pm O(\epsilon)}{\delta^2}.$$
\item \label{Nice4'} The expected cumulative value starting from $v$
is   \hfill (Corollary~\ref{cl:SW})
$$\cv(v)\in \Theta \left(\frac{v\pm \epsilon}{\delta^2}\right).$$
\item \label{Nice5'} The expected time to reach a value of $1$ or
  higher conditioned on reaching that value before $0$ is \hfill (special case of Equation~\eqref{Nice2} in Lemma~\ref{cl:nice})
$$\Ex[v]{\tau | V_{\tau} \ge 1} = \frac{[1-v^2 - c_3(1-v)](1 \pm O(\epsilon))}{3\delta ^2}.$$
\end{enumerate}

\noindent
The proofs of the above facts follow from standard martingale arguments that we outline here.

\noindent
For (\ref{Nice1'}): $\tau$ is a stopping time and $V_t$ is a martingale in $[0, \tau - 1]$, so by the optional stopping
theorem
$$\Ex[v]{V_{\tau}}  = \Ex[v]{V_0} = v =\Prx[v]{V_{\tau} \ge 1} \Ex[v]{V_{\tau} | V_{\tau}\ge 1} + (1- \Prx[v]{V_{\tau}\ge 1}) \Ex[v]{V_{\tau} | V_{\tau}< 1} .$$
which yields this fact since  $\Ex[v]{V_{\tau} | V_{\tau}\ge 1} \in [1, 1+\epsilon)$,
and, $\Ex[v]{V_{\tau} | V_{\tau}< 1} = \Ex[v]{V_{\tau} | V_{\tau}\le 0} \in (-\epsilon, 0]$.

\vspace{0.1in}
\noindent
For (\ref{Nice2'}):
Use the fact that $X_t = \frac{V_t^2}{\delta^2} - t$ is a martingale, and optional stopping at $\tau$, which yields
$$\frac{v^2}{\delta ^2} = \Ex[v]{X_{0}} = \Ex[v]{X_{\tau}}= \frac{\Ex[v]{V_{\tau}^2}}{\delta^2} - \Ex[v]{\tau}.$$
An application of (\ref{Nice1'}) completes the argument.

\vspace{0.1in}
\noindent
For  (\ref{Nice3'}) and  (\ref{Nice4'}): follow the proofs given in section 4 (with the appropriate modification for handling the boundary error).

\vspace{0.1in}
\noindent
For (\ref{Nice5'}) we slightly generalize Exercise 17.1 from~\citet{levin2009markov}.

Let $Y_t = \frac{V_t}{{\delta}}.$ Then, for $t \in [0, \tau -1]$, $Y_t$ is a martingale. In addition, if we define $c_3' = (c_3 / \delta^{3})$, then
$M_t = Y_t^3 - 3tY_t - tc_3'$ is also a martingale.
To verify this fact, let $\tDelta_t = Y_{t+1}- Y_t$. 
\begin{align*}
M_{t+1} - M_t &= \left[(Y_t+ \tDelta_t)^3 -3(t+1)(Y_t + \tDelta_t) - (t+1) c_3' \right] - (Y_t^3 -3tY_t - tc_3')\\
&= 3Y_t^2\tDelta_t + 3Y_t \tDelta_t^2  + \tDelta_t^3 - 3t \tDelta_t - 3Y_t -3\tDelta_t- c_3'.
\end{align*}
Hence
\begin{align*}
\Ex{M_{t+1} - M_t| V_t}
&= 3Y_t^2\Ex{\tDelta_t |V_t}+ 3\frac{V_t}{{\delta}}\Ex{ \tDelta_t^2| V_t}  + \Ex{\tDelta_t^3|V_t} - 3t \Ex{\tDelta_t|V_t} - 3\frac{ V_t}{{\delta}}- c_3'\\
\intertext{and since $\tDelta_t =\frac{\Delta_t}{{\delta}}$, this equals}
 &= 3\frac{V_t}{{\delta}} + \frac{c_3}{\delta^{3}} -3\frac{ V_t}{{\delta}}- c_3'=0.
\end{align*}
Finally, we apply the Optional Stopping Theorem and get
\begin{align*}
\Ex[v]{M_{0}} = \frac{v^3}{\delta^{3}} &= \Ex[v]{M_{\tau}} = 
\Ex[v]{\frac{V_{\tau}^3}{\delta^{3}}- 3 \tau \frac{V_{\tau}}{{\delta}}
- c_3' \tau}
\intertext{or equivalently}
v^3 &= \Ex[v]{V_{\tau}^3 - 3 \delta^2\tau V_{\tau}
- c_3 \tau}.\quad\quad (*)\\
\intertext{Rewrite}
 \Ex[v]{V_{\tau}^3} &=  \Ex[v]{V_{\tau}^3| V_{\tau} \ge 1} 
\Prx[v]{V_{\tau} \ge 1} + \Ex[v]{V_{\tau}^3| V_{\tau} \le 0 } \Prx[v]{V_{\tau} \le 0} \\
&= \Prx[v]{V_{\tau} \ge 1}(1+ O(\epsilon))\\
\intertext{and}
 \Ex[v]{\tau V_{\tau}} &=  \Ex[v]{\tau V_{\tau}| V_{\tau} \ge 1} 
\Prx[v]{V_{\tau} \ge 1} + \Ex[v]{\tau V_{\tau}| V_{\tau} \le 0 } \Prx[v]{V_{\tau} \le 0}. \\
&= \Ex[v]{\tau | V_{\tau} \ge 1}  \Prx[v]{V_{\tau} \ge 1} (1 + O(\epsilon))\\
\intertext{Plugging back into (*), dividing both sides by
$\Prx[v]{V_{\tau}\ge 1} = v(1 \pm O(\epsilon))$ and using fact (\ref{Nice2'}), we obtain} 
v^2 &=\left[1 - 3 \delta^2 \Ex[v]{\tau | V_{\tau} \ge 1} - c_3 (1-v) \right](1\pm O(\epsilon))
\end{align*}
Rearranging gives the result.

\vspace{0.2in}
With the above facts in hand, the other results in the paper for the random walk model follow in this more general setting mutatis mutandis.

\end{document}